\documentclass[11pt]{article}
\usepackage[a4paper,margin=1in]{geometry}
\usepackage[T1]{fontenc}
\usepackage[utf8]{inputenc}
\usepackage{lmodern}
\usepackage{microtype}
\usepackage{amsmath,amssymb,amsthm}
\usepackage{booktabs}
\usepackage{enumitem}
\usepackage{xcolor}
\usepackage{hyperref}
\usepackage[nameinlink,noabbrev]{cleveref}
\usepackage{multirow}
\usepackage{aliascnt}

\hypersetup{
  colorlinks=true,
  linkcolor=blue!50!black,
  citecolor=blue!50!black,
  urlcolor=blue!50!black,
  pdftitle={Depth-Based Local Recentering of Maxmin Landmarks for Lazy Witness Persistence},
  pdfauthor={Yifan Zhang},
  pdfsubject={Landmark selection for lazy witness persistent homology},
  pdfkeywords={persistent homology, lazy witness complex, landmark selection, halfspace depth, robust subsampling, computational topology}
}

\newaliascnt{proposition}{theorem}
\newtheorem{proposition}[proposition]{Proposition}
\aliascntresetthe{proposition}
\newaliascnt{lemma}{theorem}
\newtheorem{lemma}[lemma]{Lemma}
\aliascntresetthe{lemma}
\newaliascnt{corollary}{theorem}
\newtheorem{corollary}[corollary]{Corollary}
\aliascntresetthe{corollary}
\newaliascnt{definition}{theorem}
\newtheorem{definition}[definition]{Definition}
\aliascntresetthe{definition}
\newaliascnt{remark}{theorem}
\newtheorem{remark}[remark]{Remark}
\aliascntresetthe{remark}

\crefname{theorem}{theorem}{theorems}
\Crefname{theorem}{Theorem}{Theorems}
\crefname{proposition}{proposition}{propositions}
\Crefname{proposition}{Proposition}{Propositions}
\crefname{lemma}{lemma}{lemmas}
\Crefname{lemma}{Lemma}{Lemmas}
\crefname{corollary}{corollary}{corollaries}
\Crefname{corollary}{Corollary}{Corollaries}
\crefname{definition}{definition}{definitions}
\Crefname{definition}{Definition}{Definitions}
\crefname{remark}{remark}{remarks}
\Crefname{remark}{Remark}{Remarks}

\newcommand{\R}{\mathbb{R}}
\newcommand{\hd}{\operatorname{depth}}
\newcommand{\conv}{\operatorname{conv}}
\newcommand{\dist}{\operatorname{dist}}

\title{Depth-Based Local Recentering of Maxmin Landmarks for Lazy Witness Persistence}
\author{Yifan Zhang\\[0.5ex]
\small Department of Applied Mathematics, VSB -- Technical University of Ostrava\\
\small Department of Mathematics, University of Ostrava\\
\small Department of Algebra, Charles University\\
\small \texttt{yifan.zhang@vsb.cz}}
\date{}

\begin{document}
\maketitle

\begin{abstract}
We study a local post-processing rule for maxmin (farthest-point) landmark selection in lazy witness persistence. Maxmin landmarks induce nearest-landmark cells. Within each cell, the landmark is moved toward a data point of maximum empirical halfspace depth, with smaller moves in sparsely populated cells, and is then projected back to a data point in the same cell. The resulting support-weighted rule preserves one landmark per original maxmin cell.

We prove a convex-core guarantee for the depth target and explicit worst-case cover bounds for the recentered landmarks. Experiments on planar synthetic loops and 120 MPEG-7 (Moving Picture Experts Group 7) silhouettes separate average landmark fit from extremal coverage. On the synthetic benchmark, recentering reduces mean nearest-landmark distance in all 900 matched trials; after aggregation over landmark budgets, the mean change is $-0.0138$ across 300 dataset/noise/seed blocks, while maximum cover increases by $0.0067$ on average and a thresholded first homology ($H_1$) count is unchanged. On MPEG-7, the image-level mean distance decreases for all 120 silhouettes, while maximum cover increases for 118. Thus the correction consistently improves average landmark representativeness but does not improve worst-case coverage; on the controlled synthetic benchmark, this trade-off occurs without a detectable change in the coarse $H_1$ summary.
\end{abstract}


\section{Introduction}
Witness complexes reduce the size of a persistent-homology computation by using a small landmark set as vertices while retaining the full point cloud as witnesses \cite{deSilvaCarlsson2004,OtterEtAl2017}. The lazy witness complex further determines higher-dimensional simplices from its one-skeleton. Both constructions therefore depend on how the landmarks are selected. A common choice is maxmin, or farthest-point, sampling: after an initial seed is fixed, each new landmark is the point farthest from those already selected. Maxmin provides broad coverage, but a landmark chosen to serve an extremal point can be a poor representative of the surrounding data, especially under clutter or nonuniform sampling.

This paper considers a local correction rather than a new global sampling scheme. Maxmin landmarks define nearest-seed cells. In each cell, we select a data point of maximum empirical halfspace depth and move the seed part of the way toward it. The movement is reduced for cells with little support, measured by the number of points assigned to the cell, and the intermediate point is projected back to the same cell. The final landmarks therefore remain a subset of the input cloud and preserve the one-representative-per-cell interpretation.

The experiments distinguish two geometric quantities. The \emph{mean nearest-landmark distance} measures average representativeness, whereas the \emph{maximum cover radius} measures the worst-served signal point. These are the primary empirical outcomes. On the synthetic families, where the intended number of loops is fixed by construction, we also record whether the number of sufficiently persistent first homology ($H_1$) bars is correct. The MPEG-7 silhouettes are used only as a geometric benchmark; we do not assign them an intrinsic persistent-homology label. The principal empirical finding is a consistent reduction in mean nearest-landmark distance together with a small deterioration in maximum cover radius.

\paragraph{Contributions.}
The main contributions are:
\begin{enumerate}[leftmargin=1.5em]
    \item a cellwise depth-based post-processing family for maxmin landmarks, with support-weighted partial recentering as the main variant;
    \item a convex-core guarantee for the depth target and worst-case cover bounds for recentered landmarks; and
    \item paired experiments that separate average landmark fit from worst-case cover, with thresholded $H_1$ validation on synthetic loop families and a geometric generalization check on MPEG-7 silhouettes.
\end{enumerate}
\section{Related work and scope}
\subsection{Witness complexes and landmark geometry}
De Silva and Carlsson introduced witness complexes and the lazy witness variant as economical constructions in which the number and placement of landmarks control much of the computational cost \cite{deSilvaCarlsson2004}. Standard references on computational topology and topological data analysis discuss the same trade-off between geometric approximation and complex size \cite{EdelsbrunnerHarer2010,Carlsson2009,OtterEtAl2017}.

Later work gives conditions under which witness-type constructions recover or approximate other complexes. Guibas and Oudot study reconstruction with witness complexes, Attali et al. analyze alpha--beta witness complexes, and Dey, Fan, and Wang introduce graph-induced complexes \cite{GuibasOudot2008,AttaliEHM2007,DeyFanWang2015}. Arafat, Basu, and Bressan relate lazy witness complexes to Vietoris--Rips complexes, the clique complexes generated by thresholded pairwise distances, when the landmarks form an $\varepsilon$-net, i.e. a metric covering set at scale $\varepsilon$ \cite{ArafatBasuBressan2019}. Herick, Joachim, and Vahrenhold study adaptive approximation of persistent homology in a broader setting \cite{HerickJoachimVahrenhold2024}. These results motivate treating landmark geometry as part of the construction rather than as an incidental preprocessing choice.

\subsection{Landmark selection and robust subsampling}
Robust landmark selection for persistent homology has been studied from several directions. Stolz proposes outlier-robust persistent-homology landmark rules and compares them with random and maxmin landmarks \cite{Stolz2023}. Chazal et al. study subsampling for persistent homology from a statistical perspective, while Brunson and Skaf introduce the \emph{lastfirst} sampler and compare it with maxmin in homology-detection tasks \cite{ChazalEtAl2015,BrunsonSkaf2022}. The method considered here differs from these global samplers: it first forms a maxmin landmark set and then changes only the representative assigned to each maxmin cell.

\subsection{Halfspace depth and centerpoints}
Halfspace depth provides the local target for recentering. B\'ar\'any and Sober\'on survey the surrounding robust-geometric theory \cite{BaranySoberon2018}. Exact planar centerpoints can be computed in linear time \cite{JadhavMukhopadhyay1994}, and more general algorithms for depth-based representative points and approximate centerpoints are available in fixed dimension \cite{AgarwalSharirWelzl2008,MulzerWerner2013}. Our implementation uses a simpler discrete choice: among the data points in each cell, it selects one of maximum empirical halfspace depth. This keeps every final landmark in the input cloud.

\subsection{Silhouette data and scope of the comparison}
Persistent summaries have also been used for shape analysis, including the persistent homology transform and statistical analyses of landmark-based shape data \cite{TurnerMukherjeeBoyer2014,GambleHeo2010}. The MPEG-7 silhouette archive is a standard source of planar shapes for contour-based recognition studies \cite{BelongieMalikPuzicha2002,LateckiLakemperEckhardt2000}. We use it only as a collection of varied planar outlines, not as a shape-classification benchmark.

The comparison therefore addresses a narrow question: whether a local depth correction of maxmin can improve average landmark representativeness while retaining explicit cover bounds and without materially changing the coarse topological summary used on the controlled synthetic data.
\section{Preliminaries and notation}
\subsection{Point clouds, cells, and cover radius}
Let $X=\{x_1,\dots,x_n\}\subset \R^d$ be a finite point cloud and $S=\{s_1,\dots,s_m\}\subset X$ a set of seeds, also called landmarks. Assign each $x\in X$ to one of its nearest seeds, using a fixed deterministic rule to break distance ties, and let $P_i$ be the set of points assigned to $s_i$. The cells $P_1,\dots,P_m$ then form a partition of $X$. The \emph{cover radius} of $S$ on $X$ is $r_S(X):=\max_{x\in X}d(x,S)=\max_{x\in X}\min_{1\le j\le m}d(x,s_j)$. For maxmin seeds we write $r$ when the cloud and seed set are clear.

\subsection{Halfspace depth and deep representatives}
\begin{definition}[Halfspace depth]
For a finite set $P\subset \R^d$ and a point $y\in \R^d$, the \emph{halfspace depth} of $y$ with respect to $P$ is
\[
\hd_P(y):=\min\bigl\{|H\cap P|:H\text{ is a closed halfspace in }\R^d\text{ containing }y\bigr\}.
\]
A point $y$ is a \emph{centerpoint} of $P$ if $\hd_P(y)\ge \lceil |P|/(d+1)\rceil$.
\end{definition}

A \emph{deep representative} of a cell is either a continuous centerpoint or, as in our subset-valued implementation, a cell point of maximum empirical halfspace depth. Thus, for each cell $P_i$, we choose a \emph{deepest data point} $a_i\in\arg\max_{p\in P_i}\hd_{P_i}(p)$. In the planar experiments, depths are obtained by exhaustive evaluation of the candidate points in each cell.

\subsection{Lazy witness filtration}
Let $L=\{\ell_1,\dots,\ell_m\}$ be the landmarks and $W=\{w_1,\dots,w_N\}$ the witnesses. Write $D(a,i)=d(\ell_a,w_i)$ for the landmark--witness distance matrix. For a witness parameter $\nu\ge0$, let $m_i$ be the $\nu$-th smallest landmark distance from witness $w_i$, with $m_i=0$ when $\nu=0$. At scale $R\ge0$, the de Silva--Carlsson rule inserts an edge $[ab]$ if $\max(D(a,i),D(b,i))\le R+m_i$ for at least one witness $w_i$ \cite{deSilvaCarlsson2004}. The \emph{lazy witness complex} is the clique complex of this graph. We study persistent first homology, denoted $H_1$, of this filtration.

\section{Depth-based recentering}
All variants begin with the same maxmin landmark set $S=\{s_1,\dots,s_m\}$ and its nearest-seed cells $P_1,\dots,P_m$. Starting from an initial seed, maxmin adds at each step a point farthest from the current landmark set. For each cell $P_i$, let $a_i$ be a deepest data point and let $n_i=|P_i|$. Full recentering replaces $s_i$ by $a_i$; fixed-step recentering moves a fixed fraction of the way from $s_i$ to $a_i$. These two variants are used only in the ablation study.

The main method is \emph{support-weighted partial recentering}. Let $\bar n=|X|/m$ be the average cell size. Given $\alpha_{\max}\in(0,1]$ and $\tau>0$, define $\alpha_i=\alpha_{\max}\min(1,n_i/(\tau\bar n))$ and $z_i=(1-\alpha_i)s_i+\alpha_i a_i$. The corrected landmark $\ell_i$ is the point of $P_i$ nearest to $z_i$. Thus well-supported cells can use the full step $\alpha_{\max}$, while cells with fewer assigned points move less. Restricting the projection to $P_i$ keeps every final landmark in the input cloud and preserves one representative for each original maxmin cell.

\paragraph{Algorithm.}
Given a point cloud and landmark budget $m$: (i) compute maxmin seeds and their nearest-seed cells; (ii) compute a deepest data point in each cell; (iii) compute the support-weighted step $\alpha_i$ and intermediate point $z_i$; and (iv) replace $z_i$ by its nearest data point in the same cell. The resulting landmarks are then used as the vertices of the lazy witness filtration.

\section{Theory: convex-core and cover bounds}
This section establishes two local geometric properties: a sufficiently deep target cannot leave a sufficiently dominant closed convex core, and recentering preserves explicit worst-case cover bounds. These statements concern landmark geometry; they do not give a global approximation theorem for witness persistence.

\subsection{Depth prevents escape from a convex core}
The basic statement is slightly more general than the centerpoint threshold and is useful for the discrete deepest-point implementation.

\begin{lemma}[Convex-core property in \(\R^d\)]\label{lem:convex-core-general}
Let $P\subset \R^d$ be finite, let $C\subset \R^d$ be a closed convex set, and let $y\in \R^d$. If $|P\setminus C|<\hd_P(y)$, then $y\in C$.
\end{lemma}

\begin{proof}
Assume for contradiction that $y\notin C$. Since $C$ is closed and convex, the hyperplane separation theorem yields a closed halfspace $H$ such that $y\in H$ and $H\cap C=\varnothing$. Hence every point of $H\cap P$ lies in $P\setminus C$, so $|H\cap P|\le |P\setminus C|<\hd_P(y)$. But $H$ is a closed halfspace containing $y$, contradicting the definition of $\hd_P(y)$.
\end{proof}

\begin{corollary}[Centerpoint and deepest-point versions]\label{cor:centerpoint-core}
Let $P\subset \R^d$ be finite and let $C\subset \R^d$ be a closed convex set.
\begin{enumerate}[label=(\roman*)]
    \item If $|P\setminus C| < |P|/(d+1)$, then every centerpoint of $P$ lies in $C$.
    \item If $a\in P$ is a deepest data point and $|P\setminus C| < \hd_P(a)$, then $a\in C$.
\end{enumerate}
\end{corollary}

\begin{proof}
Part (i) follows from \Cref{lem:convex-core-general} and the centerpoint depth bound. Part (ii) is the same statement specialized to $y=a$.
\end{proof}

The corollary is local: it concerns the deep target chosen inside a fixed cell. It does not imply that the final partially recentered landmark lies in the same convex core, because the implemented landmark is obtained by moving only part of the way toward that target and then projecting to a sampled cell point. Closedness is needed in \Cref{lem:convex-core-general,cor:centerpoint-core}; for a nonclosed convex set, the separation argument yields only $y\in\overline{C}$.

\subsection{Geometry inherited from maxmin}
We next record the cover inherited from the original maxmin cells.

\begin{lemma}[Every cell lies in a maxmin ball]\label{lem:cell-ball}
Let $S=\{s_1,\dots,s_m\}\subset X$ be any seed set with cover radius $r=r_S(X)$, and let $P_i$ be the induced cells on $X$. Then $P_i\subseteq B(s_i,r)$ for each $i$. Consequently, $\conv(P_i)\subseteq B(s_i,r)$ as well.
\end{lemma}

\begin{proof}
If $x\in P_i$, then by definition $s_i$ is one of the nearest seeds to $x$. Hence $d(x,s_i)=d(x,S)\le r$. This proves $P_i\subseteq B(s_i,r)$. Since Euclidean balls are convex, the convex hull of $P_i$ is contained in the same ball.
\end{proof}

\begin{proposition}[Subset recentering preserves a \(2r\)-cover]\label{prop:two-r-cover}
Let $S=\{s_1,\dots,s_m\}\subset X$ have cover radius $r$, let $P_1,\dots,P_m$ be the induced cells, and choose any representative $\ell_i\in P_i$ for each cell. Then $L=\{\ell_1,\dots,\ell_m\}$ is a $2r$-cover of $X$.
\end{proposition}

\begin{proof}
Fix $x\in X$, say $x\in P_i$. By \Cref{lem:cell-ball}, both $x$ and $\ell_i$ lie in $B(s_i,r)$, so $d(x,\ell_i)\le d(x,s_i)+d(s_i,\ell_i)\le2r$. Since $x$ was arbitrary, $L$ is a $2r$-cover of $X$.
\end{proof}

\begin{remark}
\Cref{prop:two-r-cover} is a worst-case bound, not a claim that recentering preserves the original maxmin radius. It is independent of the depth calculation and of the shape of the cells.
\end{remark}

\subsection{Partial recentering and a projected cover bound}
The partial-recentering variants move inside the convex hull of a cell rather than necessarily to a data point of the cell. This leads to a stronger continuous bound and then to a discrete bound with an explicit projection term.

\begin{proposition}[Continuous partial-recentering cover bound]\label{prop:continuous-damped}
Let $S=\{s_1,\dots,s_m\}\subset X$ have cover radius $r$, let $P_1,\dots,P_m$ be the induced cells, and choose points $c_i\in \conv(P_i)$ together with coefficients $\alpha_i\in[0,\alpha_{\max}]$. Define $z_i=(1-\alpha_i)s_i+\alpha_i c_i$. Then the continuous landmark set $Z=\{z_1,\dots,z_m\}$ is a $(1+\alpha_{\max})r$-cover of $X$.
\end{proposition}

\begin{proof}
By \Cref{lem:cell-ball}, each $c_i\in \conv(P_i)$ lies in $B(s_i,r)$, so $d(s_i,c_i)\le r$ and therefore $d(s_i,z_i)=\alpha_i d(s_i,c_i)\le\alpha_{\max}r$. For $x\in P_i$, we also have $d(x,s_i)\le r$, hence $d(x,z_i)\le(1+\alpha_{\max})r$.
\end{proof}

\begin{proposition}[Projected partial-recentering cover bound]\label{prop:snapped-damped}
Under the hypotheses of \Cref{prop:continuous-damped}, let $q_i:=\dist(z_i,P_i)=\min_{x\in P_i}d(z_i,x)$ and $q:=\max_{1\le i\le m}q_i$. Choose a projected landmark $\ell_i\in P_i$ with $d(z_i,\ell_i)=q_i$ for each $i$. Then the projected landmark set $L=\{\ell_1,\dots,\ell_m\}$ is a $((1+\alpha_{\max})r+q)$-cover of $X$.
\end{proposition}

\begin{proof}
Fix $x\in P_i$. By \Cref{prop:continuous-damped}, $d(x,z_i)\le(1+\alpha_{\max})r$ and $d(z_i,\ell_i)=q_i\le q$. Thus $d(x,\ell_i)\le(1+\alpha_{\max})r+q$.
\end{proof}

\begin{corollary}[Two universal projected bounds]\label{cor:snapped-universal}
Under the hypotheses of \Cref{prop:snapped-damped}, assume in addition that every seed belongs to its own cell, i.e.\ $s_i\in P_i$ for all $i$, as in the maxmin-based methods studied here. Then $q_i\le d(z_i,s_i)\le\alpha_i r\le\alpha_{\max}r$, so the projected landmark set is a $((1+2\alpha_{\max})r)$-cover of $X$. Because every projected landmark $\ell_i$ lies in $P_i$, \Cref{prop:two-r-cover} also applies; equivalently, $L$ is a $\min\{2,1+2\alpha_{\max}\}r$-cover of $X$.
In particular, if $\alpha_{\max}\le 1/2$, the step-size-dependent bound improves the coarse $2r$ estimate.
\end{corollary}

\begin{proof}
Since $s_i\in P_i$, the distance from $z_i$ to $P_i$ is at most $d(z_i,s_i)$. The latter is bounded by $\alpha_i r$ exactly as in the proof of \Cref{prop:continuous-damped}. Substituting $q\le \alpha_{\max}r$ into \Cref{prop:snapped-damped} gives the first claim. The second follows because any subset-valued choice $\ell_i\in P_i$ also satisfies the hypotheses of \Cref{prop:two-r-cover}.
\end{proof}

\begin{remark}
The term $q$ measures the additional displacement caused by projecting back to the sampled cell. It can be small in a dense cell and larger in a sparse one, which motivates reducing the step in weakly supported cells. For the implemented maxmin-based rules, \Cref{cor:snapped-universal} gives a bound without a density assumption.
\end{remark}

For the support-weighted rule, the same calculation gives $d(s_i,z_i)\le\alpha_{\max}\min(1,|P_i|/(\tau\bar n))r$. Thus a cell with $|P_i|\ll\bar n$ is moved less than a well-supported cell. This is a displacement control, not a separate robustness theorem.

\subsection{Computational remarks}
The implementation used in the experiments has three main stages.
\begin{enumerate}[leftmargin=1.5em]
    \item Maxmin seed selection takes $O(nm)$ distance updates.
    \item Cell assignment takes another $O(nm)$ distance pass.
    \item The discrete depth computation in a cell of size $k$ is implemented naively by checking all candidate points against all directed lines through candidate--point pairs, which is cubic in $k$ in the implementation used in the experiments.
\end{enumerate}
Thus the landmark-selection cost is $O(nm)+O(\sum_{i=1}^m|P_i|^3)$ in the implementation used in the experiments. This is practical in the present low-budget benchmark because the cells are small. A more specialized planar centerpoint or halfspace-depth routine could reduce this cost, but such an optimization is not required for the empirical comparison carried out here.

The lazy witness stage has additional computational costs. Once the landmarks are fixed, forming the landmark--witness distance matrix costs $O(mN)$ for $N$ witnesses, and computing all edge insertion values from the de Silva--Carlsson rule in the direct implementation costs $O(m^2N)$. The subsequent simplex-tree construction is data dependent, but in the present experiments $m\le 40$, so the witness side remains comfortably secondary to the repeated benchmark loop itself. This computational profile is one reason we emphasize low landmark budgets: they are the regime in which witness constructions are most attractive and in which the local correction can still be computed exactly in the plane.

The cover bounds do not imply an interleaving between the corrected lazy witness filtration and the full Vietoris--Rips filtration of the cloud. The convex-core lemma is also local to a fixed cell and does not determine a globally optimal partition. Finally, the step parameter is a geometric tuning parameter rather than the solution of a persistence optimization problem. The experiments are interpreted within these limits.

\section{Experimental design}
\subsection{Synthetic benchmark}
The synthetic benchmark uses three planar signal families: a noisy circle, two noisy circles, and a planar figure-eight. Their intended thresholded $H_1$ counts are one, two, and two, respectively. Each trial contains $300$ signal points and $20$ added points. The added points are either uniform box outliers or a compact outlier cluster. Landmark budgets are $m\in\{20,30,40\}$, and the main benchmark uses $50$ random seeds for each signal/noise setting. Thus each method is evaluated on $3\times2\times3\times50=900$ matched trials.

\subsection{MPEG-7 silhouette benchmark}
To complement the synthetic loops, we use the MPEG-7 Core Experiment Shape-1 Part B silhouette archive as a planar benchmark. A balanced subset of $120$ silhouettes is used, with $10$ images from each of $12$ classes. Each image is converted to a planar point cloud by sampling its foreground boundary. Three versions are considered: the clean sampled boundary, the same sample with a compact outlier cluster, and the same sample with uniform planar clutter. The benchmark uses the same landmark budgets $m\in\{20,30,40\}$ and $20$ random seeds per image and setting, giving $120\times3\times3\times20=21{,}600$ matched trials per method.

The MPEG-7 archive supplies shape silhouettes rather than persistent-homology ground truth. We therefore use this experiment only to test whether the geometric behavior observed on the synthetic data persists on varied planar outlines. No silhouette-level $H_1$ accuracy claim is made.

\subsection{Methods compared}
The main benchmark compares four landmark procedures: \emph{maxmin}; the proposed \emph{support-weighted partial recentering}; the budget-matched covering comparator labelled ``$\varepsilon$-net matched'' in the companion implementation; and \emph{dense-core + maxmin}. The dense-core baseline ranks points by a local $k$-neighborhood density score with $k=8$, retains the densest $80\%$ of the cloud, and then applies maxmin. The $\varepsilon$-net label refers only to this implementation-specific comparator and is not a claim about $\varepsilon$-net algorithms in general. The ablation study also includes full and fixed-step recentering.

Unless otherwise stated, the lazy witness filtration uses $\nu=1$ and maximal scale $R_{\max}=2.1$. The main synthetic benchmark uses $(\alpha_{\max},\tau)=(0.6,1.0)$. The parameter sweep ranks $(0.6,0.5)$ highest under its stated ordering criterion, and the MPEG-7 benchmark uses that setting. Because the two main benchmarks use different parameter pairs, their numerical results are interpreted separately rather than pooled as evaluations of an identical parameter setting.

\subsection{Persistence pipeline and reference quantities}
We compute the lazy witness filtration directly from the de Silva--Carlsson definition. For witnesses $w_i$ and landmarks $\ell_a$, let $D(a,i)=d(\ell_a,w_i)$ and let $m_i$ be the witness offset determined by $\nu$. The edge insertion value is
\[
E(a,b)=\min_i \bigl(\max(D(a,i),D(b,i))-m_i\bigr)_+,
\]
where $(t)_+=\max\{t,0\}$. Vertices, edges, and triangles are inserted into a simplex tree with these filtration values, and first-homology persistence is computed with the GUDHI software library \cite{MariaEtAl2014,gudhi:urm}.

For the synthetic families, the topological summary is the number of finite $H_1$ bars with lifetime at least $0.25$. This count is compared with the known target of one loop for the circle and two loops for the two-circle and figure-eight families; the resulting binary indicator is called \emph{thresholded $H_1$-count agreement}. The geometric quantities, used on both benchmarks, are the mean and maximum distances from uncontaminated signal points to their nearest landmark. The former measures average representativeness; the latter is the empirical cover radius on the signal points. In the ablation study, \emph{outlier landmarks} denotes selected landmarks that belong to the added contamination points.

\subsection{Pairing and statistical analysis}
All method comparisons are matched within a trial: the methods use the same underlying signal or silhouette realization, contamination, landmark budget, and seed index. Raw trial counts are useful descriptively, but repeated budgets and repeated trials from the same underlying object are not treated as independent inferential units. For the main synthetic benchmark, we first average paired method differences over the three landmark budgets within each dataset/noise/seed block, leaving $3\times2\times50=300$ blocks. For MPEG-7, we average paired differences over noise type, landmark budget, and seed within each silhouette, leaving $120$ image-level units.

For continuous quantities, we report the mean paired difference, a percentile bootstrap $95\%$ confidence interval based on $20{,}000$ resamples of these blocks or images, and a two-sided paired Wilcoxon signed-rank test on the same aggregated differences. Bootstrap resampling uses a fixed random-number-generator seed of $2026$; zero differences are omitted from the signed-rank calculation. Synthetic thresholded $H_1$ agreement is analyzed at the same block level. This cluster-level analysis avoids treating all $900$ or $21{,}600$ raw matched trials as independent observations.

The archived result bundles and supporting scripts are deposited in a Zenodo dataset record \cite{ZhangZenodo2026}; the persistent link is given in the Data availability statement.

\section{Ablation study}
A 20-trial ablation study compared maxmin with full recentering, support-weighted recentering, the matched $\varepsilon$-net comparator, and dense-core + maxmin. Table~\ref{tab:ablation-aggregate-v12} summarizes the aggregate results. Full recentering gives a larger reduction in mean nearest-landmark distance than the support-weighted rule but a slightly lower thresholded $H_1$-count agreement. Dense-core + maxmin reduces the average distance most strongly and nearly eliminates outlier landmarks, but its agreement drops substantially. This comparison supports retaining the depth-based correction direction while moderating displacement in weakly supported cells.

\begin{table}[ht]
\centering
\caption{Aggregate results from the 20-trial ablation study. Lower is better for mean nearest-landmark distance and outlier landmarks; higher is better for thresholded $H_1$-count agreement.}
\label{tab:ablation-aggregate-v12}
\small
\begin{tabular}{lccc}
\toprule
Method & Agreement & Mean distance & Outlier landmarks \\
\midrule
Maxmin & 0.9889 & 0.1472 & 9.20 \\
Full recentered & 0.9861 & 0.1291 & 8.79 \\
Support-weighted partial & 0.9917 & 0.1341 & 8.99 \\
$\varepsilon$-net matched & 0.9861 & 0.1245 & 8.32 \\
Dense-core + maxmin & 0.9167 & 0.0982 & 0.92 \\
\bottomrule
\end{tabular}
\end{table}

The ablation does not establish that the support-weighted rule is globally optimal. Its role is narrower: among the tested depth-based corrections, it retains most of the mean-distance improvement while avoiding the small loss in the thresholded $H_1$ summary seen with full recentering.

\section{Experimental results}
The empirical section has three components: a 50-seed synthetic benchmark with additional baselines, a parameter study, and the MPEG-7 silhouette benchmark. The main comparisons are reported both as raw aggregate summaries and, for inference, at the block or image level defined above.

\subsection{Main synthetic benchmark with additional baselines}
Table~\ref{tab:synthetic-aggregate-v8} aggregates the 50-seed synthetic benchmark. Support-weighted partial recentering matches maxmin in aggregate thresholded $H_1$-count agreement while lowering mean nearest-landmark distance. The $\varepsilon$-net baseline provides the closest geometric comparison among the additional methods.

\begin{table}[ht]
\centering
\caption{Aggregate synthetic benchmark results over all raw matched trials. Agreement is with the known thresholded $H_1$ target. Lower is better for mean distance and maximum cover.}
\label{tab:synthetic-aggregate-v8}
\small
\begin{tabular}{lccc}
\toprule
Method & Agreement & \shortstack{Mean\\distance} & \shortstack{Maximum\\cover} \\
\midrule
Maxmin & 0.9922 & 0.1477 & 0.3477 \\
Support-weighted & 0.9922 & 0.1338 & 0.3544 \\
$\varepsilon$-net matched & 0.9900 & 0.1253 & 0.3471 \\
Dense-core + maxmin & 0.9389 & 0.0982 & 0.3485 \\
\bottomrule
\end{tabular}
\end{table}

The paired comparison against maxmin is the central empirical test. At the raw trial level, support-weighted partial recentering lowers the mean nearest-landmark distance in all $900/900$ matched trials. At the more conservative block level ($300$ dataset/noise/seed blocks), the mean paired change is $-0.01382$, with a bootstrap $95\%$ confidence interval of $[-0.01465,-0.01304]$; all $300$ block means favor recentering and the paired Wilcoxon test gives $p=6.1\times10^{-51}$. Thresholded $H_1$-count agreement is identical in aggregate ($0.9922$ for both methods), and the block-level agreement difference is centered at zero ($p=1.0$).

The maximum cover radius moves in the opposite direction. Its raw aggregate mean increases from $0.3477$ to $0.3544$. At block level the mean paired increase is $0.00666$, with a $95\%$ confidence interval of $[0.00418,0.00892]$; maximum cover is larger after recentering in $218$ of $300$ blocks ($p=2.1\times10^{-10}$). Thus the repeatable effect in this benchmark is a reduction in average signal-to-landmark distance, accompanied by a small loss in extremal coverage while the thresholded $H_1$ count remains unchanged in aggregate.

The additional baselines provide useful reference points. The $\varepsilon$-net matched method has a lower mean nearest-landmark distance than support-weighted partial recentering ($0.1253$ versus $0.1338$), a maximum cover radius close to maxmin, and slightly lower aggregate agreement ($0.9900$ versus $0.9922$). Dense-core + maxmin is more aggressive: it gives the smallest mean nearest-landmark distance, but its maximum cover radius does not improve and its agreement drops to $0.9389$. These comparisons show that support-weighted recentering is not the strongest geometric sampler in the benchmark; its role is a conservative local modification of maxmin.

\subsection{Breakdown by signal family and noise model}
Aggregate tables can hide whether the observed gain is concentrated in one easy corner of the benchmark. Table~\ref{tab:synthetic-breakdown-v12} therefore resolves the synthetic benchmark by signal family and contamination model, still averaging over landmark budgets. On the geometric side, the pattern is uniform across the reported settings: support-weighted partial recentering improves mean nearest-landmark distance in every row, with relative gains ranging from $7.7\%$ on two circles with clustered outliers to $10.0\%$ on the figure-eight with uniform clutter. In four of the six aggregated settings, the thresholded $H_1$-count agreement is exactly tied with maxmin. The remaining two cases go in opposite directions: a slight gain for the figure-eight under uniform clutter and a slight loss for two circles under uniform clutter. The most delicate case is therefore a multi-loop configuration in which the thresholded topology depends on preserving enough boundary support under diffuse contamination.

\begin{table}[ht]
\centering
\caption{Synthetic benchmark by signal family and noise model, averaged over the three landmark budgets. The last three columns report mean nearest-landmark distance and its relative reduction under support-weighted recentering.}
\label{tab:synthetic-breakdown-v12}
\scriptsize
\setlength{\tabcolsep}{3pt}
\begin{tabular}{llccccc}
\toprule
Dataset & Noise & \shortstack{Maxmin\\agreement} & \shortstack{Support-\\weighted\\agreement} & \shortstack{Maxmin\\mean distance} & \shortstack{Support-weighted\\mean distance} & \shortstack{Mean-distance\\gain} \\
\midrule
Circle & Cluster & 1.0000 & 1.0000 & 0.0827 & 0.0745 & 9.9\% \\
Circle & Uniform & 1.0000 & 1.0000 & 0.1566 & 0.1412 & 9.8\% \\
Figure-eight & Cluster & 1.0000 & 1.0000 & 0.1197 & 0.1088 & 9.1\% \\
Figure-eight & Uniform & 0.9733 & 0.9800 & 0.2003 & 0.1802 & 10.0\% \\
Two circles & Cluster & 1.0000 & 1.0000 & 0.1227 & 0.1133 & 7.7\% \\
Two circles & Uniform & 0.9800 & 0.9733 & 0.2040 & 0.1851 & 9.3\% \\
\bottomrule
\end{tabular}
\end{table}

\subsection{Dependence on landmark budget}
Witness-based landmark selection is most useful at low landmark budgets, so we also resolve the mean-distance effect by $m$. Table~\ref{tab:budget-breakdown-v12} shows a reduction at every tested budget on both benchmarks. On the synthetic benchmark, the relative improvement ranges from $9.0\%$ to $10.1\%$ across $m=20,30,40$; on MPEG-7 it ranges from $7.5\%$ to $9.2\%$. The effect is therefore not confined to a single landmark budget.

\begin{table}[ht]
\centering
\caption{Budget dependence of mean nearest-landmark distance for maxmin and support-weighted partial recentering.}
\label{tab:budget-breakdown-v12}
\small
\begin{tabular}{llccc}
\toprule
Benchmark & $m$ & Maxmin & Support-weighted & Relative reduction \\
\midrule
\multirow{3}{*}{Synthetic}
 & 20 & 0.2131 & 0.1933 & 9.3\% \\
 & 30 & 0.1339 & 0.1219 & 9.0\% \\
 & 40 & 0.0961 & 0.0863 & 10.1\% \\
\midrule
\multirow{3}{*}{MPEG-7}
 & 20 & 0.1011 & 0.0917 & 9.2\% \\
 & 30 & 0.0661 & 0.0608 & 8.0\% \\
 & 40 & 0.0476 & 0.0440 & 7.5\% \\
\bottomrule
\end{tabular}
\end{table}

\subsection{Parameter sensitivity of support-weighted partial recentering}
We also checked sensitivity to the two parameters of the support-weighted rule. The support-weighted partial method has two parameters, $\alpha_{\max}\in(0,1]$ and $\tau>0$. The sweep uses $\alpha_{\max}\in\{0.3,0.5,0.6,0.8\}$ and $\tau\in\{0.5,1.0,1.5\}$ on representative circle and two-circle settings under clustered and uniform contamination.

Across the $12$ tested parameter pairs, mean thresholded $H_1$-count agreement ranges from $0.98125$ to $0.98750$, while mean nearest-landmark distance ranges from $0.1499$ to $0.1622$. The larger variation is associated with $\alpha_{\max}$: small values move the seeds less and therefore yield smaller geometric changes. The parameter $\tau$ has a smaller effect on this grid. The pair $(\alpha_{\max},\tau)=(0.6,0.5)$ attains the highest observed agreement together with the lowest mean distance among the settings tied at that agreement level, and is used in the MPEG-7 benchmark. This is an empirical choice on the tested grid, not an optimization result.

\begin{table}[ht]
\centering
\caption{Top parameter settings for support-weighted partial recentering in the parameter sweep, ranked by thresholded $H_1$-count agreement and then by mean nearest-landmark distance.}
\label{tab:parameter-top-v8}
\small
\begin{tabular}{ccccc}
\toprule
$\alpha_{\max}$ & $\tau$ & Agreement & \shortstack{Mean\\distance} & \shortstack{Outlier\\landmarks} \\
\midrule
0.6 & 0.5 & 0.9875 & 0.1499 & 8.11 \\
0.6 & 1.0 & 0.9875 & 0.1516 & 8.19 \\
0.6 & 1.5 & 0.9875 & 0.1542 & 8.26 \\
0.8 & 0.5 & 0.9813 & 0.1478 & 8.08 \\
\bottomrule
\end{tabular}
\end{table}

\subsection{MPEG-7 silhouette case study}
We next test the geometric comparison on silhouettes that were not constructed for the present study. The MPEG-7 case study uses a balanced subset of $120$ silhouettes, with $10$ shapes from each of the apple, bat, bell, bird, bottle, butterfly, camel, car, cup, fish, hammer, and turtle classes, together with clean, clustered-noise, and uniform-noise variants.

Table~\ref{tab:silhouette-based-v8} summarizes the raw aggregate results. Support-weighted partial recentering lowers the mean nearest-landmark distance from $0.0716$ to $0.0655$, an $8.5\%$ reduction. At image level, the mean paired change is $-0.00606$, with a bootstrap $95\%$ confidence interval of $[-0.00656,-0.00555]$; all $120$ image means favor recentering ($p=2.0\times10^{-21}$). The maximum cover radius moves in the opposite direction, increasing from $0.1722$ to $0.1797$. At image level the mean paired increase is $0.00752$, with a $95\%$ confidence interval of $[0.00685,0.00817]$; $118$ of $120$ image means favor maxmin ($p=2.6\times10^{-21}$).

\begin{table}[ht]
\centering
\caption{Aggregate MPEG-7 silhouette geometry over all raw matched trials. Lower is better.}
\label{tab:silhouette-based-v8}
\small
\begin{tabular}{lcc}
\toprule
Method & Mean distance & Maximum cover \\
\midrule
Maxmin & 0.0716 & 0.1722 \\
Support-weighted & 0.0655 & 0.1797 \\
$\varepsilon$-net matched & 0.0629 & 0.1716 \\
Dense-core + maxmin & 0.0536 & 0.2208 \\
\bottomrule
\end{tabular}
\end{table}

The additional baselines make the trade-off clearer. The $\varepsilon$-net matched method has lower mean distance and a slightly smaller maximum cover radius than support-weighted recentering. Dense-core + maxmin gives the smallest mean distance but the largest maximum cover radius. The proposed correction is therefore less aggressive than dense filtering, but it is not the strongest geometric sampler among the methods tested.

\subsection{Noise-wise breakdown on MPEG-7}
Table~\ref{tab:silhouette-based-noise-v12} resolves the mean-distance comparison by contamination model. Clean and clustered silhouettes behave similarly, whereas uniform clutter is the hardest regime for all methods. Support-weighted partial recentering improves mean nearest-landmark distance over maxmin in all three noise classes, by about $6.9\%$ on clean and clustered silhouettes and by $10.3\%$ under uniform clutter. The $\varepsilon$-net matched baseline remains the strongest geometric comparator, especially under clutter.

\begin{table}[ht]
\centering
\caption{Noise-wise mean nearest-landmark distance on the MPEG-7 silhouette benchmark. Lower is better.}
\label{tab:silhouette-based-noise-v12}
\small
\begin{tabular}{lccc}
\toprule
Noise & Maxmin & Support-weighted & $\varepsilon$-net matched \\
\midrule
Clean & 0.0536 & 0.0499 & 0.0488 \\
Cluster & 0.0629 & 0.0585 & 0.0542 \\
Uniform & 0.0983 & 0.0882 & 0.0857 \\
\bottomrule
\end{tabular}
\end{table}

\subsection{Runtime note}
The stored per-method timing fields are too coarse for a scaling analysis, so we report only the end-to-end wall-clock time of the silhouette benchmark. On a 16-core desktop, the balanced MPEG-7 run of $120\times3\times3\times20=21{,}600$ matched trials per method completed in 177 minutes 12 seconds using process-level parallelism. This number documents the scale of the experiment; it is not used as a per-method complexity comparison.

\subsection{Interpretation of the experimental evidence}
Across both planar benchmarks, support-weighted partial recentering lowers the mean nearest-landmark distance relative to maxmin, whereas the maximum cover radius is usually slightly larger. This average-versus-extremal trade-off is the most reproducible result of the study. On the synthetic benchmark, where the target loop count is fixed by construction, the thresholded $H_1$ summary remains essentially unchanged.

The budget-matched $\varepsilon$-net comparator is at least as competitive geometrically in these experiments. Dense-core filtering illustrates a different trade-off: it reduces average distance aggressively but can remove boundary information and lower the synthetic thresholded $H_1$ agreement. Support-weighted recentering is therefore a local modification of maxmin rather than a replacement for witness-specific sampling strategies.

\section{Discussion and limitations}
The experiments separate average from extremal landmark geometry. Maxmin greedily targets an extremal covering criterion, whereas recentering moves selected seeds toward deeper parts of their cells. The same move can therefore reduce average signal-to-landmark distance while increasing the distance of the worst-served points. Both main benchmarks show this trade-off, and the cover bounds in Section~4 are consistent with it: they limit deterioration but do not imply numerical improvement of the cover radius.

The topological evidence is more limited than the geometric evidence. It is evaluated only on the controlled synthetic families, where the intended number of loops is fixed by construction. There, thresholded $H_1$ agreement is unchanged in aggregate. The MPEG-7 silhouettes are used only for geometry because the benchmark does not provide an intrinsic persistent-homology target. No general persistence advantage is claimed.

Support-weighted recentering is most useful when maxmin already distributes landmarks adequately but some selected seeds are locally unrepresentative because of clutter or uneven sampling. It can be less suitable when the relevant topology depends strongly on extremal boundary points. The ablation results are consistent with this limitation: full recentering can move too far inward, whereas cell-size weighting reduces that movement.

The present implementation also has practical limits. Planar halfspace depth is computed by an exhaustive routine and is not optimized for large cells or high ambient dimension. The parameters are selected from a finite sweep rather than from an optimization theorem. The supplied result tables record the signal size, contamination count, persistence threshold, and principal algorithmic parameters, but they do not encode every synthetic generator constant. Exact reproduction therefore depends on the companion implementation rather than on the archived result tables alone. These restrictions limit the scope of the study but do not affect the paired geometric comparisons computed directly from the recorded signal and landmark distances.

\section{Conclusion}
We studied a depth-based local correction of maxmin landmarks for lazy witness persistence. The support-weighted variant moves each maxmin seed toward a deepest point of its own cell and reduces the step in cells with fewer assigned points. A convex-core lemma justifies the depth-based target under a local contamination condition, while the cover propositions bound the possible loss of worst-case coverage after recentering.

Across the synthetic and MPEG-7 experiments, the correction consistently reduces mean nearest-landmark distance, while the maximum cover radius is slightly worse on average. In the synthetic benchmark, the thresholded $H_1$ count remains essentially unchanged relative to maxmin. The method is therefore a simple adjustment of average landmark representativeness, not a replacement for maxmin's extremal covering behavior and not a general improvement theorem for witness persistence.

\section*{Data availability}
The experimental record is available at \url{https://doi.org/10.5281/zenodo.19676292}. The result bundles used in this study contain the planar synthetic benchmark, the parameter-sensitivity study, and the MPEG-7 silhouette benchmark. Reproduction of the full data-generation pipeline requires the companion implementation associated with that record.

\section*{Funding}
This work is partially funded by VSB -- Technical University of Ostrava, grant no. SP2026/019. Yifan Zhang is co-funded by the European Union under the REFRESH--Research Excellence For REgion Sustainability and High-tech Industries project number CZ.10.03.01/00/22\_003/0000048 via the Operational Programme Just Transition.

\end{document}